%% file: main.tex
\documentclass[11pt,letterpaper]{article}

\usepackage{amsthm,color,latexsym,graphicx,url}
\usepackage[margin=1in]{geometry}
\usepackage[small,labelfont=bf]{caption}
\usepackage[labelformat=simple]{subcaption}

\usepackage{libertine}\usepackage[libertine]{newtxmath}
\usepackage[scaled=0.96]{zi4}
\usepackage[utf8]{inputenc}

\usepackage[table]{xcolor}

\title{Complexity of Reconfiguration in Surface Chemical Reaction Networks} 

\newif\ifabstract
\abstracttrue
\newif\iffull
\ifabstract \fullfalse \else \fulltrue \fi

\usepackage{hyperref}

{\makeatletter \hypersetup{pdftitle={\@title}}}

{\makeatletter
 \gdef\xxxmark{%
   \expandafter\ifx\csname @mpargs\endcsname\relax 
     \expandafter\ifx\csname @captype\endcsname\relax 
       \marginpar{xxx}
     \else
       xxx 
     \fi
   \else
     xxx 
   \fi}
 \gdef\xxx{\@ifnextchar[\xxx@lab\xxx@nolab}
 \long\gdef\xxx@lab[#1]#2{\textbf{[\xxxmark #2 ---{\sc #1}]}}
 \long\gdef\xxx@nolab#1{\textbf{[\xxxmark #1]}}
}

{\makeatletter \gdef\fps@figure{!htbp}}

\let\realbfseries=\bfseries
\def\bfseries{\realbfseries\boldmath}

\newtheorem{theorem}{Theorem}
\newtheorem{lemma}[theorem]{Lemma}
\newtheorem{corollary}[theorem]{Corollary}
\theoremstyle{definition}

\let\epsilon=\varepsilon

\begin{document}

\author{Robert M. Alaniz, Josh Brunner, Michael Coulombe, Erik D. Demaine, Jenny Diomidova, \\
 Timothy Gomez, Elise Grizzell, Ryan Knobel, Jayson Lynch, Andrew Rodriguez,  \\ 
 Robert Schweller, Tim Wylie}










\maketitle

\input{abstract.tex}

\input{intro.tex}

\input{table.tex}

\input{scrn_model.tex} 

\input{swap_reactions.tex}

\input{burnout.tex}

\input{single_reactions.tex}

\input{conclusion.tex}

\bibliographystyle{abbrv}
\bibliography{main.bib}

\appendix

\end{document}

%% file: abstract.tex
\begin{abstract}

We analyze the computational complexity of basic reconfiguration problems for the recently introduced surface Chemical Reaction Networks (sCRNs), where ordered pairs of adjacent species nondeterministically transform into a different ordered pair of species according to a predefined set of allowed transition rules (chemical reactions).
In particular, two questions that are fundamental to the simulation of sCRNs are whether a given configuration of molecules can ever transform into another given configuration, and whether a given cell can ever contain a given species, given a set of transition rules. 
We show that these problems can be solved in polynomial time,
are NP-complete, or are PSPACE-complete in a variety of different settings, including when adjacent species just swap instead of arbitrary transformation (swap sCRNs), and when cells can change species a limited number of times ($k$-burnout). Most problems turn out to be at least NP-hard except with very few distinct species (2 or 3).

\end{abstract}

%% file: intro.tex
\section{Introduction}
The ability to engineer molecules to perform complex tasks is an essential goal of molecular programming.  A popular theoretical model for investigating molecular systems and distributed systems is Chemical Reaction Networks (CRNs) \cite{chen2014deterministic, soloveichik2008computation}. The model abstracts chemical reactions to independent rule-based interactions that creates a mathematical framework equivalent \cite{cook2009programmability} to other well-studied models such as Vector Addition Systems \cite{Karp:1969:JCSS} and Petri nets \cite{Petri1962PHD}. CRNs are also interesting for experimental molecular programmers, as examples have been built using DNA strand displacement (DSD) \cite{soloveichik2010dna}.

Abstract Surface Chemical Reaction Networks (sCRNs) were introduced in \cite{parScale} as a way to model chemical reactions that take place on a surface, where the geometry of the surface is used to assist with computation. In this work, the authors gave a possible implementation of the model similar to ideas of spatially organized DNA circuits \cite{muscat2013dna}. This strategy involves DNA strands being anchored to a DNA origami surface. These strands allow for ``species'' to be attached. Fuel complexes are pumped into the system, which perform the reactions. While these reactions are more complex than what has been implemented in current lab work, it shows a route to building these types of networks. 

\subsection{Motivation}
Feed-Forward circuits using DNA hairpins anchored to a DNA origami surface were implemented in \cite{chatterjee2017spatially}. This experiment used a single type of fuel strand. The copies of the fuel strand attached to the hairpins and were able to drive forward the computation. 

A similar model was proposed in \cite{dannenberg2015dna}, which modeled DNA walkers moving along tracks. These tracks have guards that can be opened or closed at the start of computation by including or omitting specific DNA species at the start. DNA walkers have provided interesting implementations such as robots that sort cargo on a surface \cite{thubagere2017cargo}.

A new variant of surface CRNs we introduce is the $k$-burnout model in which cells can switch states at most $k$ time before being stuck in their final state.  This models the practical scenario in which state changes expend some form of limited fuel to induce the state change.  Specific experimental examples of this type of limitation can be seen when species encode ``fire-once" DNA strand replacement reactions on the surface of DNA origami, as is done within the Signal Passing Tile Model~\cite{HSA:2012:PLS}.


\subsection{Previous Work}
The initial paper on sCRNs \cite{parScale} gave a 1D reversible Turing machine as an example of the computational power of the model. They also provided other interesting constructions such as building dynamic patterns, simulating continuously active Boolean logic circuits, and cellular automata. 
Later work in \cite{progSim} gave a simulator of the model, improved some results of \cite{parScale}, and gave many open problems- some of which we answer here.

 In \cite{swap}, the authors introduce the concept of swap reactions. These are reversible reactions that only ``swap'' the positions of the two species. The authors of \cite{swap} gave a way to build feed-forward circuits using only a constant number of species and reactions. These swap reactions may have a simpler implementation and also have the advantage of the reverse reaction being the same as the forward reaction, which makes it possible to reuse fuel species. 

A similar idea for swap reactions on a surface that has been studied theoretically are friends-and-strangers graphs \cite{defant2021friends}. This model was originally introduced to generalize problems such as the 15 Puzzle and Token Swapping. In the model, there is a location graph containing uniquely labeled tokens and a friends graph with a vertex for every token, and an edge if they are allowed to swap locations when adjacent in the location graph. The token swapping problem can be represented with a complete friends graph, and the 15 puzzle has a grid graph as the location graph and a star as the friends graph (the `empty square' can swap with any other square). Swap sCRNs can be described as multiplicities friends-and-strangers graph \cite{milojevic2022connectivity}, which relax the unique restriction, with the surface grid (in our case the square grid) as the location graph and the allowed reactions forming the edges of the friends graph. 

\subsection{Our Contributions}
In this work, we focus on two main problems related to sCRNs. The first is the reconfiguration problem, which asks given two configurations and a set of reactions, can the first configuration be transformed to the second using the set of reactions. The second is the $1$-reconfiguration problem, which asks whether a given cell can ever contain a given species.  Our results are summarized in Table~\ref{tab:results}. The first row of the table comes from the Turing machine simulation in \cite{parScale} although it is not explicitly stated. The size comes from the smallest known universal reversible Turing machine \cite{morita2007universal} (see \cite{woods2009complexity} for a survey on small universal Turing machines.)

We first investigate swap reactions in Section \ref{sec:swap}. We prove both problems are PSPACE-complete using only four species and three swap reactions. For reconfiguration, we show this complexity is tight by showing with three or less species and only swap reactions the problem is in P.

In Section \ref{sec:burn}, we study a restriction on surface CRNs called $k$-burnout where each species is guaranteed to only transition $k$ times. This is similar to the freezing restriction from Cellular Automata \cite{goles2021complexity,goles2015introducing,theyssier2022freezing} and Tile Automata \cite{chalk2018freezing}. We start with a simple reduction showing reconfiguration is NP-complete in $2$-burnout. This is also of interest since the reduction only uses three species types and a reaction set of size one. For $1$-reconfiguration, we show the problem is also NP-complete in $1$-burnout sCRNs. This reduction uses a constant number of species. 

In Section \ref{sec:1react}, we analyze reconfiguration for all sCRNs that have a reaction set of size one. For the case of only two species, we show for every possible reaction, the problem is solvable in polynomial time. With three species or greater, we show that reconfiguration is NP-complete. The hardness comes from the reduction in burnout sCRNs. 

Finally, in Section \ref{sec:conc}, we conclude the paper by discussing the results as well as many open questions and other possible directions for future research related to surface CRNs.

%% file: table.tex
\definecolor{header}{rgb}{0.29,0,0.51} 
\definecolor{gray}{rgb}{0.85,0.85,0.85}
\definecolor{blueheader}{HTML}{2F4C74}
\def\header#1{\cellcolor{header}\textcolor{white}{\textbf{#1}}}

\begin{table}[t]
    \centering
    \rowcolors{2}{gray!75}{white}
\begin{tabular}{| c | c | c  | c | c | c | c |}
 \hline \arrayrulecolor{white}
 \header{Problem} & \header{Type} & \header{Graph} & \header{Species} & \header{Rules} & \header{Result} & \header{Ref} \\
 \arrayrulecolor{black} \hline
 Reconfiguration & sCRN & 1D & $17$ & $67$ & PSPACE-complete & \cite{parScale} \\ 
  1-Reconfiguration & Swap sCRN & Grid & $4$ & $ 3$ & PSPACE-complete & Thm.~\ref{thm:oneRecSwap} \\ 
   $1$-Reconfiguration & Swap sCRN & Any &   $\leq 3$ & Any & P & Thm.~\ref{thm:1recEzSwap}\\ 
 $1$-Reconfiguration & Swap sCRN & Any & Any & $\leq 2$ & P & Thm.~\ref{thm:1recEzSwap}\\
 Reconfiguration & Swap sCRN & Grid & $4$ & $ 3$ & PSPACE-complete & Thm.~\ref{thm:swapReduc} \\ 
 Reconfiguration & Swap sCRN & Any & $\leq 3$ & Any & P & Thm.~\ref{thm:easySwaps}\\ 
 Reconfiguration & Swap sCRN & Any & Any & $\leq 2$ & P & Thm.~\ref{thm:easySwaps}\\
 Reconfiguration & $2$-burnout & Grid & $3$ & $1$ & NP-complete & Thm.~\ref{thm:hamPathRed} \\ 
 1-Reconfiguration & $1$-burnout & Grid & $17$ & 40 & NP-complete & Thm.~\ref{thm:3satburn} \\ 
  Reconfiguration & sCRN & Grid & $\geq 3$ & $1$ &  NP-complete & Cor.~\ref{cor:s3r1} \\ 
 Reconfiguration & sCRN & Any & $ \leq 2$ & $1$ & P & Thm.~\ref{thm:2s1r} \\ \hline
\end{tabular}
    \caption{Summary of our and known complexity results for sCRN reconfiguration problems, depending on the type of sCRN, number of species, and number of rules. All problems are contained in PSPACE, while all $k$-burnout problems are in NP.}
    \label{tab:results}
\end{table}

%% file: scrn_model.tex
\section{Surface CRN model}

\textbf{Chemical Reaction Network.}
A \emph{chemical reaction network (CRN)} is a pair $\Gamma = (S, R)$ where $S$ is a set of species and $R$ is a set of reactions, each of the form $A_1 + \cdots + A_j \to B_1 + \cdots + B_k$ where $A_i,B_i \in S$.
(We do not define the dynamics of general CRNs, as we do not need them here.)

\textbf{Surface, Cell, and Species.}
A \emph{surface} for a CRN $\Gamma$ is an (infinite) undirected graph $G$.
The vertices of the surface are called \emph{cells}.
A \emph{configuration} is a mapping from each cell to a species from the set $S$. 
While our algorithmic results apply to general surfaces,
our hardness constructions assume the practical case where
$G$ is a grid graph, i.e., an induced subgraph of the infinite square grid
(where omitted vertices naturally correspond to cells without any species).
%
When $G$ is an infinite graph, we assume there is some periodic pattern of cells that is repeated on the edges of the surface. 
Figure \ref{fig:scrnModel} shows an example set of species and reactions and a configuration of a surface. 

\textbf{Reaction.}
A \emph{surface Chemical Reaction Network (sCRN)} consists of a surface and a CRN, where every \emph{reaction} is of the form $A + B \rightarrow C + D$ denoting that, when $A$ and $B$ are in neighboring cells, they can be replaced with $C$ and $D$. $A$ is replaced with $C$ and $B$ with~$D$. 

\textbf{Reachable Configurations.}
For two configurations $I, T$, we write $I \rightarrow^1_\Gamma T$ if there exists a $r \in R$ such that performing reaction $r$ on a pair of species in $I$ yields the configuration $T$. Let $I \rightarrow_\Gamma T$ be the transitive closure of $I \rightarrow^1_\Gamma T$, including loops from each configuration to itself. Let $\Pi(\Gamma, I)$ be the set of all configurations $T$ for which $I \rightarrow_\Gamma T$ is true. A sequence of reachable states is shown in Figure \ref{fig:configs_example}



\begin{figure}[t]
    \centering
    \captionsetup{justification=centering}
    \includegraphics[width=0.6\textwidth]{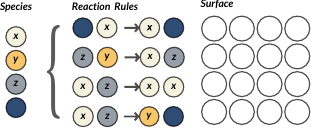}
    \caption{Example sCRN system.}
    \label{fig:scrnModel}
\end{figure}

\begin{figure}
    \centering
    \includegraphics[width=\textwidth]{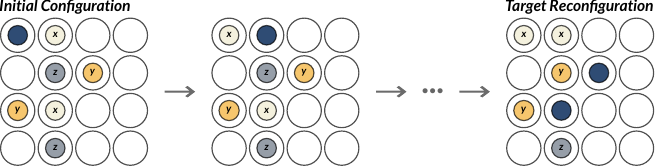}
    \caption{An initial, single step, and target configurations}
    \label{fig:configs_example}
\end{figure}
\subsection{Restrictions}

\textbf{Reversible Reactions.}
A set of reactions $R$ is \emph{reversible} if, for every rule $A + B \rightarrow C + D$ in $R$, the reaction $C + D \rightarrow A + B$ is also in $R$. We may also denote this as a single reversible reaction $A + B \rightleftharpoons C + D$.


\textbf{Swap Reactions.}
A reaction of the form $A + B \rightleftharpoons B + A$ is called a \emph{swap reaction}.


\textbf{\boldmath $k$-Burnout.} 
In the $k$-burnout variant of the model, each vertex of the system's graph can only switch states at most $k$ times (before ``burning out" and being stuck in its final state).








\subsection{Problems}

\textbf{Reconfiguration Problem.}
Given a sCRN $\Gamma$ and two configurations $I$ and $T$, is $T \in \Pi(\Gamma, S)$?

\textbf{\boldmath $1$-Reconfiguration Problem.} 
Given a sCRN $\Gamma$, a configuration $I$, a vertex $v$, and a species $s$, does there exist a $T \in \Pi(\Gamma, S)$ such that $T$ has species $s$ at vertex $v$?








%% file: swap_reactions.tex
\section{Swap Reactions}\label{sec:swap}
In this section, we will show $1$-reconfiguration and reconfiguration with swap reactions is PSPACE-complete with only $4$ species and $3$ swaps in Theorems \ref{thm:oneRecSwap} and \ref{thm:swapReduc}. We continue by showing that this complexity is tight, that is, reconfiguration with $3$ species and swap reactions is tractable in Theorems \ref{thm:easySwaps} and \ref{thm:1recEzSwap}. 

\subsection{Reconfiguration is PSPACE-complete}
We prove PSPACE-completeness by reducing from the motion planning through gadgets framework introduced in \cite{demaine2018computational}. This is a one player game where the goal is to navigate a robot through a system of gadgets to reach a goal location. The problem of changing the state of the entire system to a desired state has been shown to be PSPACE-complete \cite{ani2022traversability}. This reduction treats the model as a game where the player must perform reactions moving a robot species through the surface.

\subsubsection*{The Gadgets Framework}

\textbf{Framework.} A gadget is a finite set of locations and a finite set of states. Each state is a directed graph on the locations of the gadgets, describing the \emph{traversals} of the gadget. An example can be seen in Figure~\ref{fig:L2T}. Each edge (traversal) describes a move the robot can take in the gadget and what state the gadget ends up in if the robot takes that traversal. A robot enters from the start of the edge and leaves at the exit. 

In a \emph{system} of gadgets there are multiple gadgets connected by their locations.  The \emph{configuration} of a system of gadgets is the state of all gadgets in the system. There is a single robot that starts at a specified location. The robot is allowed to move between connected locations and allowed to move along traversals within gadgets. The system of gadgets can also be restricted to be planar, in which case the cyclic order of the locations on the gadgets is fixed, and the gadgets along with their connections must be embeddable in the plane without crossings.

The \emph{1-player motion planning reachability problem} asks whether there exists a sequence of moves within a system of gadgets which takes the robot from its initial location to a target location. The \emph{1-player motion planning reconfiguration problem} asks whether there exists a sequence of moves which brings the configuration of a system of gadgets to some target configuration.

There are many sets of motion planning models and gadgets to build our reduction. We select 1-player over 0-player since in the sCRN model there are many reactions that may occur and we are asking whether there exists a sequence of reactions which reaches some target configuration; in the same way 1-player motion planning asks if there exists a sequence of moves which takes the robot to the target location. The existential query of possible moves/swaps remains the same regardless of whether a player is making decisions vs them occurring by natural processes. The complexity of the gadgets used here are considered in the 0-player setting in \cite{demaine2022pspace}.


\begin{figure}[t]
    \centering
        \includegraphics[width=0.5\linewidth]{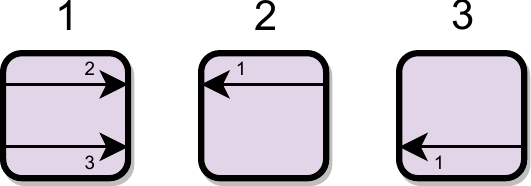}
    \caption{The Locking 2-Toggle (L2T) gadget and its states from the motion planning framework. The numbers above indicate the state and when a traversal happens across the arrows, the gadget changes to the indicated state.}
    \label{fig:L2T}    
\end{figure}

\textbf{Locking 2-Toggle.} The Locking $2$-toggle (L2T) is a $4$ location, $3$ state gadget. The states of the gadget are shown in Figure \ref{fig:L2T}. The L2T has advantages because it universal for reversible deterministic gadgets. Reversibility was important to picking a gadget since swap reactions are naturally reversible. 


\subsubsection*{Constructing the L2T}


We will show how to simulate the L2T in a swap sCRN system. Planar 1-player motion planning with the L2T was shown to be PSPACE-complete \cite{demaine2018computational}. We now describe this construction.

\textbf{Species.} We utilize $4$ species types in this reduction and we name each of them according to their role. First we have the \emph{wire}. The wire is used to create the connection graph between gadgets and can only swap with the robot species. The \emph{robot} species is what moves between gadgets by swapping with the wire and represents the robot in the framework. Each gadget initially contains $2$ robot species, and there is one species that starts at the initial location of the robot in the system. The robot can also swap with the key species. Each gadget has exactly $1$ \emph{key} species. The key species is what performs the traversal of the gadget by swapping with the lock species. The \emph{lock} species can only swap with the key. There are $4$ locks in each gadget. The locks ensure that only legal traversals are possible by the robot species.

These species are arranged into gadgets consisting of two length-$5$ horizontal tunnels. The two tunnels are connected by a length-$3$ central vertical tunnel at their $3$rd cell. At the $4$th cell of both tunnels there is an additional degree $1$ cell connected we will call the holding cell. 

\begin{figure}
    \centering
    \begin{subfigure}[b]{0.31\textwidth}
        \centering
        \includegraphics[scale=0.6]{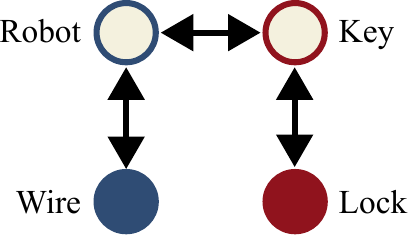}
        \caption{Swap rules/species}
        \label{fig:swapDiag}
    \end{subfigure}
    \begin{subfigure}[b]{0.2\textwidth}
        \includegraphics[scale=0.74]{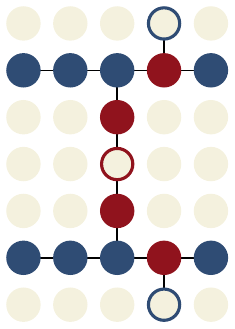}
        \caption{State 1}
        \label{fig:gadgetS1}
    \end{subfigure}\hfill
    \begin{subfigure}[b]{0.2\textwidth}
        \includegraphics[scale=0.74]{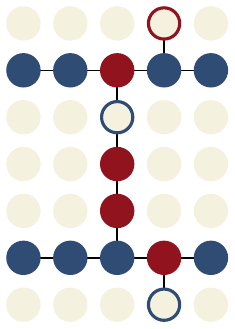}
        \caption{State 2}
        \label{fig:gadgetS2}
    \end{subfigure}\hfill
    \begin{subfigure}[b]{0.2\textwidth}
        \includegraphics[scale=0.74]{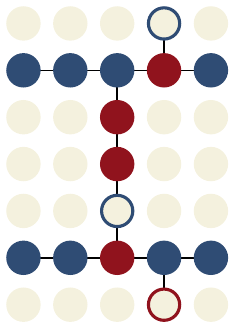}
        \caption{State 3}
        \label{fig:gadgetS3}
    \end{subfigure}
    \caption{Locking 2-toggle implemented by swap rules. (a) The swap rules and species names. (b-d) The three states of the locking 2-toggle.}
    \label{fig:L2Tgadget}
\end{figure}

\textbf{States and Traversals.} The states of the gadget we build are represented by the location of the key species in each gadget. If the key is in the central tunnel of the gadget then we are in state $1$ as shown in Figure \ref{fig:gadgetS1}. Note that in this state the key may swap with the adjacent locks, however we consider these configurations to also be in state 1 and take advantage of this later. The horizontal tunnels of the gadget in this state contain a single lock with an adjacent robot species.

States $2$ and $3$ are reflections of each other (Figures \ref{fig:gadgetS2} and \ref{fig:gadgetS3}). This state has a robot in the central tunnel and the key in the respective holding cell. The gadget in this state can only be traversed from right to left in one of the tunnels. 

\begin{figure}
    \centering
    \begin{subfigure}[b]{0.13\textwidth}
        \includegraphics[width=\linewidth]{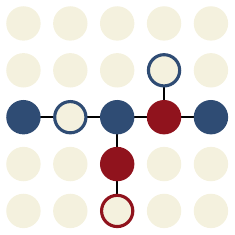}
    \end{subfigure}\hfill
    \begin{subfigure}[b]{0.13\textwidth}
        \includegraphics[width=\linewidth]{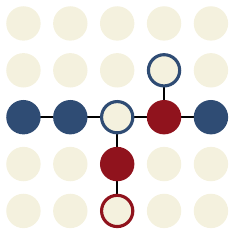}
    \end{subfigure}\hfill
    \begin{subfigure}[b]{0.13\textwidth}
        \includegraphics[width=\linewidth]{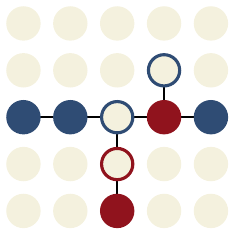}
    \end{subfigure}\hfill
    \begin{subfigure}[b]{0.13\textwidth}
        \includegraphics[width=\linewidth]{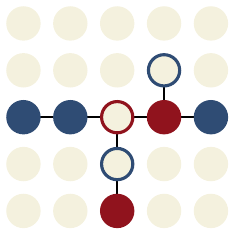}
    \end{subfigure}\hfill
    \begin{subfigure}[b]{0.13\textwidth}
        \includegraphics[width=\linewidth]{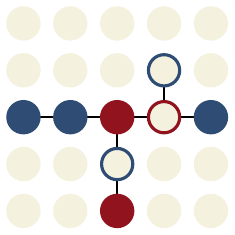}
    \end{subfigure}\hfill
    \begin{subfigure}[b]{0.13\textwidth}
        \includegraphics[width=\linewidth]{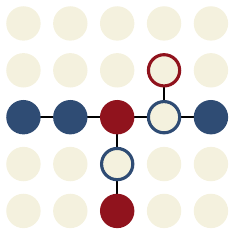}
    \end{subfigure}\hfill
    \begin{subfigure}[b]{0.13\textwidth}
        \includegraphics[width=\linewidth]{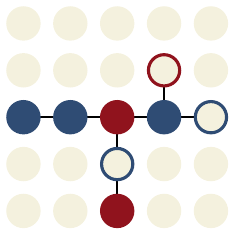}
    \end{subfigure}
    \caption{Traversal of the robot species.} 
    \label{fig:traversal}
\end{figure}

Figure \ref{fig:traversal} shows the process of a robot species traversing through the gadget. Notice when a robot species ``traverse'' a gadget, it actually traps itself to free another robot at the exit. We prove two lemmas to help verify the correctness of our construction. The lemmas prove the gadgets we design correctly implement the allowed traversals of a locking 2-toggle.

\begin{lemma}\label{lem:rightward}
    A robot may perform a rightward traversal of a gadget through the north/south tunnel if and only if the key is moved from the central tunnel to the north/south holding cell. 
\end{lemma}
\begin{proof}
    The horizontal tunnels in state $1$ allow for a rightward traversal. The robot swaps with wires until it reaches the third cell where it is adjacent to two locks. However the key in the central tunnel may swap with the locks to reach the robot. The key and robot then swap. The key is then in the horizontal tunnel and can swap to the right with the lock there. It may then swap with the robot in the holding cell. This robot then may continue forward to the right and the key is stuck in the holding cell.

    Notice when entering from the left the robot will always reach  a cell adjacent to lock species. The robot may not swap with locks so it cannot traverse unless the key is in the central tunnel. 
\end{proof}

\begin{lemma}\label{lem:leftward}
    A robot may perform a leftward traversal of a gadget through the north/south tunnel if and only if the key is moved from the north/south holding cell to the central tunnel.
\end{lemma}
\begin{proof}
    In state $2$ the upper tunnel can be traversed and in state $3$ the lower tunnel can be traversed. The swap sequence for a leftward traversal is the reverse of the rightward traversal, meaning we are undoing the swaps to return to state $1$. The robot enters the gadget and swaps with the key, which swaps with the locks to move adjacent to the central tunnel. The key then returns to the central tunnel by swapping with the robot. The robot species can then leave the gadget to the left. 

    A robot entering from the right will not be able to swap to the position adjacent to the holding cell if it contains a lock. This is true in both tunnels in state $1$ and in the non-traversable tunnels in states $2$ and $3$. 
\end{proof}

We use these lemmas to first prove PSPACE-completeness of 1-reconfiguration. We reduce from the planar 1-player motion planning reachability problem.

\begin{theorem}\label{thm:oneRecSwap}
$1$-reconfiguration is PSPACE-complete with $4$ species and $3$ swap reactions or greater even when the surface is a subset of the grid graph. 
\end{theorem}
\begin{proof}
    Given a system of gadgets create a surface encoding the connection graph between the locations. Each gadget is built as described above in a state representing the initial state of the system. Ports are connected using multiple cells containing wire species. When more than two ports are connected we use degree-3 cells with wire species. The target cell for $1$-reconfiguration is a cell containing a wire located at the target location in the system of gadgets.  

    If there exists a solution to the robot reachability problem then we can convert the sequence of gadget traversals to a sequence of swaps. The swaps relocate a robot species to the location as in the system of gadgets. 

    If there exists a swap sequence to place a robot species in the target cell there exists a solution to the robot reachability problem. Any swap sequence either moves an robot along a wire, or traverses it through a gadget. From Lemmas \ref{lem:rightward} and \ref{lem:leftward} we know the only way to traverse a gadget is to change its state (the location of its key) and a gadget can only be traversed in the correct state. 
\end{proof}

Now we show Reconfiguration in sCRNs is hard with the same set of swaps is PSPACE-complete as well. We do so by reducing from the Targeted Reconfiguration problem which asks, given an initial and target configuration of a system of gadgets, does there exist sequence of gadget traversals to change the state of the system from the initial to the target and has the robot reach a target location. Note prior work only shows reconfiguration (without specifying the robot location) is PSPACE-complete\cite{ani2022traversability} however a quick inspection of the proof of Theorem 4.1 shows the robot ends up at the initial location so requiring a target location does not change the computational complexity for the locking 2-toggle. One may also find it useful to note that the technique used in \cite{ani2022traversability} for gadgets and in \cite{hearn2009games} for Nondeterministic Constraint Logic can be applied to reversible deterministic systems more generally. This means the method described in those could be used to give an alternate reduction directly from 1-reconfiguration of swap sCRNs to reconfiguration of swap sCRNs. 

\begin{theorem}\label{thm:swapReduc}
Reconfiguration is PSPACE-complete with $4$ species and $3$ swap reactions or greater. 
\end{theorem}
\begin{proof}
    Our initial and target configurations of the surface are built with the robot species at the robots location in the system of gadget, and each key is placed according to the starting configuration of the gadget. 

    Again as in the previous theorem we know from Lemmas \ref{lem:rightward} and \ref{lem:leftward} the robot species traversal corresponds to the traversals of the robot in the system of gadgets. The target surface can be reached if and only the target configuration in the system of gadgets is reachable. 
\end{proof}

\subsection{Polynomial-Time Algorithm}


Here we show that the previous two hardness results are tight: when restricting to a smaller cases, both problems become solvable in polynomial time. We prove this by utilizing previously known algorithms for \emph{pebble games}, where labeled pebbles are placed on a subset of nodes of a graph (with at most one pebble per node). A \emph{move} consists of moving a pebble from its current node to an adjacent empty node. These pebble games are again a type of multiplicity friends-and-strangers graph. 

\begin{theorem}\label{thm:easySwaps}
Reconfiguration is in P with $3$ or fewer species and only swap reactions.
Reconfiguration is also in P with $2$ or fewer swap reactions and any number of species.
\end{theorem}
\begin{proof}

First we will cover the case of only two swap reactions. There are two possibilities: the two reactions share a common species or they do not. If they do not, we can partition the problem into two disjoint problems, one with only the species involved in the first reaction and the other with only the species from the second reaction. Each of these subproblems has only one reaction, and is solvable if and only if each connected component of the surface has the same number of each species in the initial and target configurations. 

The only other case is where we have three species, A, B, and C, where A and C can swap, B and C can swap, but A and B cannot swap. In this case, we can model it as a pebble motion problem on a graph. Consider the graph of the surface where we put a white pebble on each A species vertex, a black pebble on each B species vertex, and leave each C species vertex empty. A legal swap in the surface CRN corresponds to sliding a pebble to an adjacent empty vertex. Goraly et al.~\cite{goraly2010multi} gives a linear-time algorithm for determining whether there is a feasible solution to this pebble motion problem. Since the pebble motion problem is exactly equivalent to the surface CRN reconfiguration problem, the solution given by their algorithm directly says whether our surface CRN problem is feasible.
\end{proof}
\begin{theorem}\label{thm:1recEzSwap}
$1$-reconfiguration is in P with $3$ or fewer species and only swap reactions.
$1$-reconfiguration is also in P with $2$ or fewer swap reactions.
\end{theorem}
\begin{proof}
If there are only two swap reactions, we again have two cases depending on whether they share a common species. If they do not share a common species, then we only need to consider the rule involving the target species. The problem is solvable if and only if the connected component of the surface of species involved in this reaction containing the target cell also has at least one copy of the target species. Equivalently, if the target species is A, and A and B can swap, then there must either be A at the target location or a path of B species from the target location to the initial location of an A species.

The remaining case is when we again have three species, A, B, and C, where A and C can swap, B and C can swap, but A and B cannot swap. If C is the target species, then the problem is always solvable as long as there is any C in the initial configuration. Otherwise, suppose without loss of generality that the target species is A. Some initial A must reach the target location. For each initial A, consider the modified problem which has only that single A and replaces all of the other copies of A with B. A sequence of swaps is legal in this modified problem if and only if it was legal in the original problem. The original problem has a solution if and only if any of the modified ones do. We then convert each of these problems to a robot motion planning problem on a graph: place the robot at the vertex with a single copy of A, and place a moveable obstacle at each vertex with a B. A legal move is either sliding the robot to an adjacent empty vertex or sliding an obstacle to an adjacent empty vertex. Papadimitriou et al. \cite{papadimitriou1994motion} give a simple polynomial time algorithm for determining whether it is possible to get the robot to a given target location. By applying their algorithm to each of these modified problems (one for each cell that has an initial A), we can determine whether any of them have a solution in polynomial time (since there are only linearly many such problems), and thus determine whether the original 1-reconfiguration problem has a solution in polynomial time.

\end{proof}




%% file: burnout.tex
\section{Burnout}\label{sec:burn}
In this section, we show reconfiguration in $2$-burnout with species $(A, B, C)$ and reaction $A + B \rightarrow C + A$ is NP-complete in Theorem \ref{thm:hamPathRed}. Next, we show $1$-reconfiguration in $1$-burnout with $17$ species and $40$ reactions is NP-complete in Theorem \ref{thm:3satburn}.

Reconfiguration and 1-Reconfiguration for burnout sCRNs are in NP since there is the length of any reconfiguration is bounded. For space we do not include this proof but note this has been proved in other system such as Resource Bounded Cellular Automata \cite{dennunzio2017computational}, Freezing Cellular Automata \cite{goles2021complexity} and Freezing Tile Automata \cite{caballero2020verification}.

\subsection{2-Burnout Reconfiguration}
This is a simple reduction from Hamiltonian Path, specifically when we have a stated start and end vertex. 

\begin{figure}[t]
    \centering
    \includegraphics[width=0.85\textwidth]{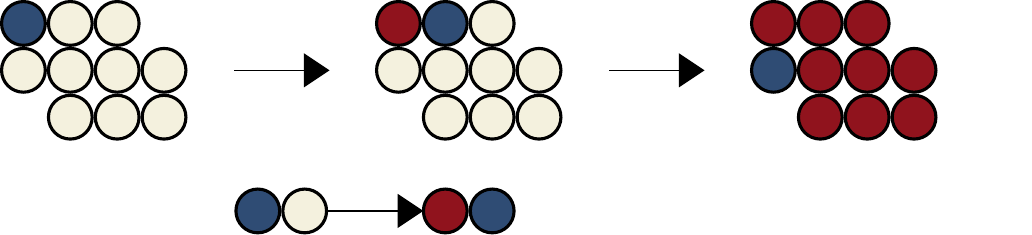}
    \caption{An example reduction from Hamiltonian Path. We are considering graphs on a grid, so any two adjacent locations are connected in the graph. Left: an initial board with the starting location in blue.  Middle: One step of the reaction. Right: The target configuration with the ending location in blue. Bottom: the single reaction rule.}
    \label{fig:hamPath}
\end{figure}

\begin{theorem}\label{thm:hamPathRed}
Reconfiguration in $2$-burnout sCRNs with species $(A, B, C)$ and reaction $A + B \rightarrow C + A$ is NP-complete even when the surface is a subset of the grid graph. It is also NP-complete with the same species and reactions without the $2$-burnout restriction. 
\label{thm:hamPath}
\end{theorem}
\begin{proof}
Let $\Gamma = \{(A, B, C), (A + B \rightarrow C + A)\}$. Given an instance of the Hamiltonian path problem on a grid graph $H$ with a specified start and target vertex $v_s$ and $v_t$, respectively, create a surface $G$ where each cell in $G$ is a node from $H$. Each cell contains the species $B$ except for the cell representing $v_s$ which contains species $A$. The target surface has species $C$ in every cell except for the final node containing $A$, $v_t$. An example can be seen in Figure \ref{fig:hamPath}.

The species $A$ can be thought of as an agent moving through the graph. The species $B$ represents a vertex that hasn't been visited yet, while the species $C$ represents one that has been. Each reaction moves the agent along the graph, marking the previous vertex as visited. 

($\Rightarrow$) If there exists a Hamiltonian path, then the target configuration is reachable. The sequence of edges in the path can be used as a reaction sequence moving the agent through the graph, changing each cell to species $C$ finishing at the cell representing $v_t$. 

($\Leftarrow$) If the target configuration is reachable, there exists a Hamiltonian path. The sequence of reactions can be used to construct the path that visits each of the vertices exactly once, ending at $v_t$. 

Note that we have not discussed the effect of Burnout on the reduction. However since each cell transitions through species in the following order: $B, A, C$ this reaction always results in a $2$-burnout sCRN so the reduction holds with and without the restriction. 

This means the CRN is $2$-burnout which bounds the max sequence length for reaching any reachable surface, putting the reconfiguration problem in NP. 
\end{proof}

\subsection{1-Burnout 1-Reconfiguration}
For $1$-burnout 1-reconfiguration, we show NP-completeness by reducing from 3SAT and utilizing the fact that once a cell has reacted it is burned out and can no longer participate in later reactions.

\begin{figure}[ht]
    \centering
    \includegraphics[scale=0.8]{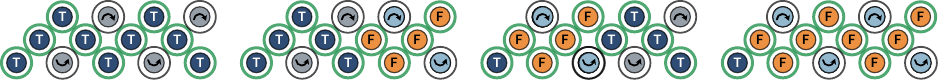}
    \caption{All the possible configurations of two variable gadgets.}
    \label{fig:settingVariables}
\end{figure}


\begin{theorem} \label{thm:3satburn}
    $1$-reconfiguration in $1$-burnout sCRNs with $17$ species and $40$ reactions is NP-complete even when the surface is a subset of the grid graph. It is also NP-complete with the same species and reactions without the $1$-burnout restriction. 
\end{theorem}

\begin{proof}
We reduce from 3SAT. The idea is to have an `agent' species traverse the surface to assign variables and check that the clauses are satisfied by `walking' through each clause. If the agent can traverse the whole surface and mark the final vertex as  `satisfied', there is a variable assignment that satisfies the original 3SAT instance.

\emph{Variable Gadget.} The variable gadget is constructed to allow for a nondeterministic assignment of the variable via the agent walk. At each intersection, the agent `chooses' a path depending on the reaction that occurs. If the agent chooses `true' for a given variable, it will walk up then walk down to the center species. If the agent chooses `false', the agent will walk down then walk up to the center species. From the center species, the agent can only continue following the path it chose until it reaches the next variable gadget. Examples of the agent assigning variables can be seen in Figure \ref{fig:settingVariables}.

\begin{figure}[t]
    \centering
    \begin{subfigure}[b]{0.4\textwidth}
        \centering
        \includegraphics[scale=0.8]{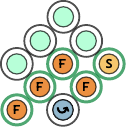}
        \caption{Successful navigation of an intersection.}
        \label{fig:var_turn}
    \end{subfigure}\hfil
    \begin{subfigure}[b]{0.5\textwidth}
        \centering
        \includegraphics[scale=0.8]{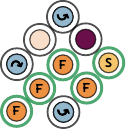}
        \caption{Agent stuck due to not following the assignment.}
        \label{fig:var_turn_block}
    \end{subfigure}
    \caption{The assignment `locking' process.}
    \label{fig:var_blocking}
\end{figure}

Each variable assignment is `locked' by way of geometric blocking. When the agent encounters a variable gadget whose variable has already been assigned, the agent must follow that same assignment or it will get `stuck' trying to react with a burnt out vertex. This can be seen in Figure \ref{fig:var_blocking}.

\emph{Initial Configuration.} First, the configuration is constructed with variable gadgets connected in a row, one for each variable in the 3SAT instance. This row of variable gadgets is where the agent will nondeterministically assign values to the variables. Next, a row of variable gadgets, one row for each clause, is placed on top of the assignment row, connected with helper species to fill in the gaps.

For each clause, if a certain variable is present, the center species of the variable gadget reflects its literal value from the clause. For example, if the variable $x1$ in clause $c1$ should be true to satisfy the clause, the variable gadget representing $x1$ in $c1$'s row will contain a $T$ species in the center cell. Lastly, the agent species is placed in the bottom left of the configuration. An example configuration can be seen in Figure \ref{fig:burnoutInput}. 

\begin{figure}[t]
\centering
    \begin{subfigure}[b]{0.49\textwidth}
        \includegraphics[width=1.0\linewidth]{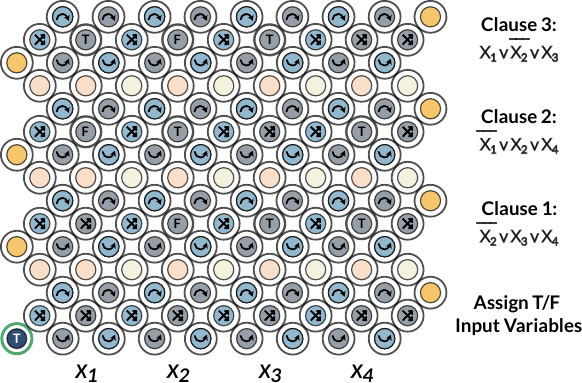} 
        \caption{Example starting configuration.}
    \end{subfigure}\hfill
    \begin{subfigure}[b]{0.49\textwidth}
        \includegraphics[width=1.0\linewidth]{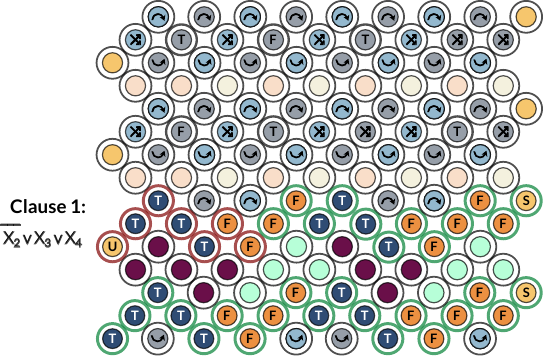}
        \caption{The surface after evaluating the first clause. }
    \end{subfigure}
    
    \caption{Reduction from 3SAT to 1-burnout 1-reconfiguration. (a) The starting configuration of the surface for the example formula $\varphi=(\lnot x_2 \lor x_3 \lor x_4)\land(\lnot x_1 \lor x_2 \lor x_4)\land(x_1 \lor \lnot x_2 \lor x_3)$. (b) The configuration after evaluating the first clause. A red outline represents the unsatisfied state, and a green outline represents the satisfied state.}
    \label{fig:burnoutInput}
\end{figure}

The agent begins walking and nondeterministically assigns a value to each variable. After assigning every variable, the agent walks right to left. If at an intersection, the agent chooses a different assignment than it did its first pass, the agent becomes `stuck' only being able to react with a burnt out vertex. 

After walking all the way to the left, the first clause can be checked. The agent starts in the unsatisfied state, walking through each variable in the row, left to right. If the current variable assignment at a variable gadget satisfies this clause, the agent changes to the satisfied state and continues walking. If the agent walks through all the variables without becoming satisfied, the computation ends. If the clause was satisfied, the agent continues by walking back, right to left, to begin evaluation of the next clause. If the agent walks all the way to the final vertex with a satisfied state, then the initial variable assignment satisfies all the clauses.

($\Rightarrow$) If there exists a variable assignment that satisfies the 3SAT instance, then the final vertex can be marked with the satisfied state $s$. The agent can only mark the final cell with the satisfied state $s$ if all clauses can be satisfied.

($\Leftarrow$) If the final vertex can be marked with satisfied state $s$, there exists a variable assignment that satisfies the 3SAT instance. The variable assignment that the agent nondeterministically chose can be read and used to satisfy the 3SAT instance.
\end{proof}

\begin{figure}[t]
    \centering
    \includegraphics[width=\textwidth]{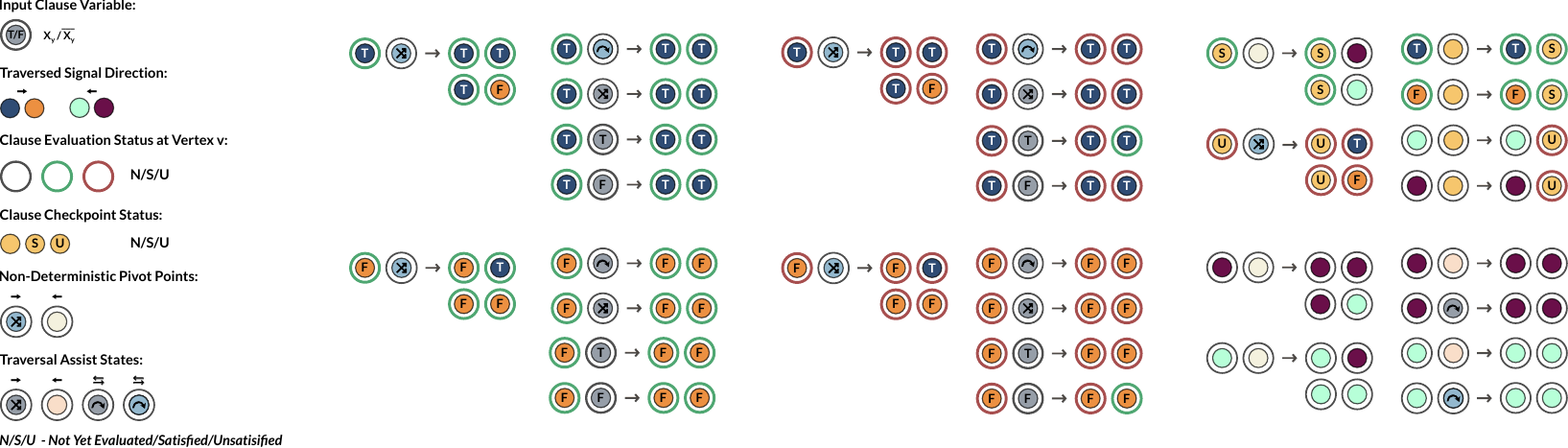}
    \caption{Species identification and transition rules for 1-burnout 1-reconfiguration.}
    \label{fig:burnoutStatesTransitions}
\end{figure}






%% file: single_reactions.tex
\section{Single Reaction}\label{sec:1react}
When limited to a single reaction, we show a complete characterization of the reconfiguration problem. There exists a reaction using $3$ species for which the problem is NP-complete. For all other cases of $1$ reaction, the problem is solvable in polynomial time. 

\subsection{2 Species}
We start with proving reconfiguration is in P when we only have $2$ species and a single reaction.  

\begin{lemma}
Reconfiguration with species $\{A, B\}$ and reaction $A + A \rightarrow A + B$ OR $A + B \rightarrow A + A$ is solvable in polynomial time on any surface. 
\label{lem:sameDiff}
\end{lemma}
\begin{proof}
The reaction $A + B \rightarrow A + A$ is the reverse of the first case. By flipping the target and initial configurations, we can reduce from reconfiguration with $A + B \rightarrow A + A$ to reconfiguration $A + A \rightarrow A + B$.

We now solve the case where we have the reaction $A + A \rightarrow A + B$. 

All cells that start and end with species $B$ can be ignored as they do not need to be changed, and can not participate in any reactions. If there is a cell that contains $B$ in the initial configuration but $A$ in the target, the instance is `no' as $B$ may never become $A$. 

Let any cell that starts in species $A$ but ends in species $B$ be called a \emph{flip} cell, and any species that starts in $A$ and stays in $A$ a \emph{catalyst} cell. 

An instance of reconfiguration with these reactions is solvable if and only if there exists a set of spanning trees, each rooted at a catalyst cell, that contain all the flip cells. Using these trees, we can construct a reaction sequence 
from post-order traversals of each spanning tree, where we have each non-root node react with its parent to change itself to a $B$.
In the other direction, given a reaction sequence, we can construct the spanning trees by pointing each flip cell to the neighbor it reacts with.
\end{proof}

\begin{lemma}
Reconfiguration with species $\{A, B\}$ and reaction $A + A \rightarrow B + B$ is solvable in polynomial time on any surface. 
\label{lem:2Sflip}
\end{lemma}
\begin{proof}
Reconfiguration in this case can be reduced to perfect matching. Create a graph $M$ including a node for each cell in $S$ containing the $A$ species initially and containing $B$ in the target, with edges between nodes of neighboring cells. If $M$ has a perfect matching, then each edge in the matching corresponds to a reaction that changes $A$ to $B$. If the target configuration is reachable, then the reactions form a perfect matching since they include each cell exactly once.
\end{proof}



\begin{theorem}\label{thm:2s1r}
Reconfiguration with $2$ species and $1$ reaction is in P on any surface.
\end{theorem}
\begin{proof}
As we only have two species and a single reaction, we can analyze each of the four cases to show membership in P.
We divide into two cases:

\textbf{\boldmath $A + A$:} When a species reacts with itself, it can either change both species, which is shown to be in P by Lemma~\ref{lem:2Sflip}; or it changes only one of the species, which is in P by Lemma~\ref{lem:sameDiff}. 

\textbf{\boldmath $A + B$:} When two different species react, they can either change to the same species, which is in P by Lemma~\ref{lem:sameDiff}; or they can both change, which is a swap and thus is in P by Theorem~\ref{thm:easySwaps}.
\end{proof}

\subsection{3 or more Species}
Moving up to $3$ species and $1$ reaction, we showed earlier that there exists a reaction for which reconfiguration is NP-complete in Theorem~\ref{thm:hamPathRed}.  
Here, we give reactions for which reconfiguration between $3$ species is in P, and in Corollary~\ref{cor:s3r1} we prove that all remaining reactions are isomorphic to one of the reactions we've analyzed.

\begin{lemma}\label{lem:s3Match}
Reconfiguration with species $(A, B, C)$ and reaction $A + B \rightarrow C + C$ is solvable in polynomial time on any surface.
\end{lemma}
\begin{proof}
    At a high level, we create a new graph of all the cells that must change to species $C$, and add an edge when the two cells can react with each other. Since a reaction changes both cells to $C$ we can think of the reaction as ``covering'' the two reacting cells. Finding a perfect matching in this new graph will give a set of edges along which to perform the reactions to reach the target configuration. 

    Consider a surface $G$ and a subgraph $G' \subseteq G$ where we include a vertex $v'$ in $G'$ for each cell that contain $A$ or $B$ in the initial configuration and $C$ in the target configuration. We include an edge $(u',v')$ between any vertices in $G'$ that contain different initial species, i.e. any pair of cell which one initially contains $A$ and the other initially $B$. 

    Reconfiguration is possible if and only if there is a perfect matching in $G'$. If there is a perfect matching then there exists a set of edges which cover each cell once. Since $G'$ represents the cells that must change states, and the edges between them are reactions, the covering can be used as a sequence of pairs of cells to react. If there is a sequence of reactions then there exists a perfect matching in $G'$: each cell only reacts once so the matching must be perfect, and the cells that react have edges between them in $G'$.
\end{proof}

\begin{lemma}\label{lem:s3r1catalyst}
Reconfiguration with species $(A, B, C)$ and reaction $A + B \rightarrow A + C$ is solvable in polynomial time on any surface.
\end{lemma}
\begin{proof}
    The instance of reconfiguration is solvable if and only if any cell that ends with species $C$ either contained $C$ in the initial configuration, or started with species $B$ and have an $A$ adjacent to perform the reaction. Additionally, since a reaction cannot cause a cell to change to $A$ or $B$, each cell with an $A$ or $B$ in the target configuration must contain the same species in the initial configuration.
\end{proof}

The final case we study is $4$ species $1$ reaction. Any sCRN with $5$ or more species and $1$ reaction has a species which is not included in the reaction. 

\begin{lemma}\label{lem:4s1r}
Reconfiguration with species $(A, B, C, D)$ and the reaction $A + B \rightarrow C + D$ is in P on any surface.
\end{lemma}
\begin{proof}
    We can reduce Reconfiguration with $A + B \rightarrow C + D$ to perfect matching similar to Lemma \ref{lem:s3Match}. Create a new graph with each vertex representing a cell in the surface that must change species. Add an edge between each pair of neighboring cells that can react (between one containing $A$ and the other $B$). A perfect matching then corresponds to a sequence of reactions that changes each of the species in each cell to $C$ or $D$. 
\end{proof}

\begin{corollary}\label{cor:s3r1}
Reconfiguration with $3$ or greater species and $1$ reaction is NP-complete on any surface. 
\end{corollary}
\begin{proof}
    First, from Theorem \ref{thm:hamPathRed} we see that there exists a case of reconfiguration with $3$ species that is NP-hard with or without the burnout restriction. 

    For membership in NP, we analyze each possible reaction. We note that we only need to consider two cases for the left hand side of the rule, $A + A$ and $A + B$. Any other reaction is isomorphic to one of this form as we can relabel the species. For example, rule $B + C \rightarrow A + A$ can be relabeled as $A + B \rightarrow C + C$. Also, we know that $C$ must appear somewhere in the right hand side of the rule. If it does not then the reaction only takes place between two species, which is always polynomial time as shown above, or it involves a species we can relabel as $C$.

    Here are the cases for $A + B$ and our analysis results:

    \begin{table}[htp]
    \centering
    \begin{tabular}{|c|c|}
    \hline
       $A + B \rightarrow A + C $  & P in Lemma \ref{lem:s3r1catalyst} \\ \hline
       $A + B \rightarrow C + B $  & P in Lemma \ref{lem:s3r1catalyst} under isomorphism \\ \hline
       $A + B \rightarrow C + A $  & NP in Theorem \ref{thm:hamPathRed} \\ \hline
       $A + B \rightarrow B + C $  & NP in Theorem \ref{thm:hamPathRed} under isomorphism \\ \hline
       $A + B \rightarrow C + C $  & P in Lemma \ref{lem:s3Match} \\ \hline
       $A + B \rightarrow C + D $  & P in Lemma \ref{lem:4s1r} \\ \hline
    \end{tabular}
    \end{table}

    When we have $A + A$ on the left side of the rule, the only case we must consider is $A + A \rightarrow B + C$ (since all $3$ species must be included in the rule). We have already solved this reaction: first swap the labels of $A$ and $C$ giving rule $C + C \rightarrow B + A$, then reverse the rule to $B + A \rightarrow C + C$ and swap the initial and target configuration. Finally since rules do not care about orientation this is equivalent to the rule $A + B \rightarrow C + C$ in Lemma \ref{lem:s3Match}.

    Finally, for $4$ species and greater, the only new case is $A + B \rightarrow C + D$, which is proven to be in P in Lemma \ref{lem:4s1r}. Any other case would have species that are not used since a rule can only have $4$ different species in it. 

    Thus, all cases are either in NP, or in P which is a subset of NP, therefore, the problem is in NP. Also, since our results for each case apply for any surface, the same is true in general.
\end{proof}

%% file: conclusion.tex
\section{Conclusion}\label{sec:conc}
In this paper, we explored the complexity of the configuration problem within natural variations of the surface CRN model.  While general reconfiguration is known to be PSPACE-complete, we showed that it is still PSPACE-complete even with several extreme constraints. We first considered the case where only swap reactions are allowed, and showed reconfiguration is PSPACE-complete with only four species and three distinct reaction types.  We further showed that this is the smallest possible number of species for which the problem is hard by providing a polynomial-time solution for three or fewer species when only using swap reactions.

We next considered surface CRNs with rules other than just swap reactions. First, we considered the burnout version of the reconfiguration problem, and then followed by the normal version with small species counts.  In the case of $2$-burnout, we showed reconfiguration is NP-complete for three species and one reaction type, and $1$-burnout is NP-complete for 17 species with 40 distinct reaction types.  Without burnout, we achieved, as a corollary, that three species, one reaction type is NP-complete while showing that dropping the species count down to two yields a polynomial-time solution.

\subsection{Computing Polynomial Space Functions}
An interpretation of Theorem \ref{thm:oneRecSwap} is that surface Chemical Reactions are capable of computing any function that can be computed in polynomial space. Perhaps the most important PSPACE-Complete is the acceptance problem for polynomial space Turing machines. While there may be a few reduction between these problems, we can may turn any polynomial space Turing machine into a surface CRN such that the robot species swaps with a wire species at a target location. In experiments one can imagine the target location as having a special type of wire species that acts as a reporting, emitting a signal when it reacts with the robot species. The size of the surface is polynomial in the space of the Turing machine since these are all polynomial time reductions. While we do not claim this experiment can be done with such a small number of species, but rather that theoretically more sequence efficient reaction systems which can compute should exists by taking advantage of the surface. 

Our polynomial time algorithms describe experiments with 1, 2, or 3 reactions on surfaces where well studied algorithms for problems such as matching and motion planning may be of use. 

\subsection{Open Problems}
This work introduced new concepts that leaves open a number of directions for future work. While we have fully characterized the complexity of reconfiguration for the swap-only version of the model, the complexity of reconfiguration with general rule types for three species systems remains open if the system uses more than one rule. All of hardness results also use a square grid graph, while our algorithms work on general surfaces. We would like to know if the threshold for hardness can be lowered on more general graphs. In the $1$-burnout variant of the model, we have shown 1-reconfiguration to be NP-complete, but the question of general reconfiguration remains a ``burning'' open question. 